\documentclass[]{interact}

\usepackage{epstopdf}
\usepackage{hyperref}

\usepackage{footmisc} 
\usepackage{subcaption}

\usepackage{mathrsfs}
\usepackage{dsfont}

\theoremstyle{plain}
\newtheorem{theorem}{Theorem}[section]
\newtheorem{lemma}[theorem]{Lemma}
\newtheorem{corollary}[theorem]{Corollary}
\newtheorem{proposition}[theorem]{Proposition}

\theoremstyle{definition}
\newtheorem{definition}[theorem]{Definition}
\newtheorem{example}[theorem]{Example}
\theoremstyle{assumption}
\newtheorem{assumption}[theorem]{Assumption}

\theoremstyle{remark}
\newtheorem{remark}{Remark}

\def\shuffle{\sqcup\mathchoice{\mkern-7mu}{\mkern-7mu}{\mkern-3.2mu}{\mkern-3.8mu}\sqcup}

\def\Sig{\mathrm{Sig}}
\usepackage{mathtools}
\usepackage{setspace}

\begin{document}

\articletype{ARTICLE TEMPLATE}

\title{Numerical method for model-free pricing of exotic derivatives using rough path signatures\footnote{Opinions expressed in this paper are those of the authors, and do not necessarily reflect the view of JP Morgan.
\\
\indent \indent The authors would like to thank Samuel Cohen and David Pr\"omel for their helpful insights towards the improvement of this paper.
\\
\indent \indent This work was supported by The Alan Turing Institute under the EPSRC grant EP/N510129/1.
}}

\author{
\name{Terry Lyons\textsuperscript{a, b}, Sina Nejad\textsuperscript{a, b} and Imanol Perez Arribas\textsuperscript{a, b, c}\thanks{CONTACT Imanol Perez Arribas. Email: imanol.perez@maths.ox.ac.uk.}}
\affil{\textsuperscript{a}Mathematical Institute, University of Oxford, UK; \textsuperscript{b}The Alan Turing Institute, London, UK; \textsuperscript{c}J.P. Morgan, London, UK}
}

\maketitle

\begin{abstract}
We estimate prices of exotic options in a discrete-time model-free setting when the trader has access to market prices of a rich enough class of exotic and vanilla options. This is achieved by estimating an unobservable quantity called ``implied expected signature'' from such market prices, which are used to price other exotic derivatives. The implied expected signature is an object that characterises the market dynamics.
\end{abstract}


\onehalfspacing

\section{Introduction}
A framework that has often been considered in model-free finance (see \cite{acciaio, obloj, discrete_modelfree1, obloj2}) is that where the market contains a family of financial derivatives written on a risky asset, which can be statically traded. Under certain no-arbitrage assumptions which forbid making profits without taking any risk, one may be interested in studying whether it is possible to derive a unique price for other financial derivatives. In a model-free framework we do not impose a probability measure on the market dynamics -- the objective is to obtain prices of potentially illiquid contracts, without assuming any particular model for the dynamics of the risky asset.

The objective of this paper is to extend the ideas proposed in \cite{sina} to discrete-time, model-free frameworks. In \cite{sina}, the authors introduce a family of basic exotic derivatives (called \textit{signature payoffs}) that are used to price (or hedge) other exotic options. In this paper, we extend this approach to obtain a numerical method to price exotic derivatives from market data in discrete time in settings with minimal assumptions about the market. We illustrate our methodology with various numerical experiments that demonstrate  the feasibility of our approach.

It is well-known (\cite{putcall}) that if prices of call and put options of all strikes and a fixed maturity are known, one can derive the price of any European contingent claim of the same maturity by writing such contracts as linear combinations of call and put options. Signature options extend the idea of approximating complex contracts in terms of simpler, basic contracts to the setting of exotic path-dependent payoffs.

Our approach is the following: first, we replicate certain exotic derivatives using a family of primitive securities -- signature payoffs. Then, assuming we have access to market prices of a rich enough family of exotic derivatives, we infer the \textit{implied expected signature} -- an unobservable quantity that characterises the entire market dynamics. Finally, we use the inferred implied expected signature to price other exotic payoffs.

Exotic derivatives are typically illiquid and cannot in general observe the market price of a path-dependent exotic derivative. However, there are multiple data providers that offer consensus market prices of a range of OTC (over-the-counter) exotic derivatives -- examples include DeriveXperts Mercure\footnote{\url{http://www.derivexperts.com/services/mercure}} and Totem Markit\footnote{\url{https://ihsmarkit.com/products/totem.html}}. These prices reflect the consensus prices from market participants. Moreover, certain  FX exotic options such as double no-touch options are liquid and we can observe their market prices. Therefore, exploring the extent to which the market prices of these exotic options contain enough information to price other exotic options is a natural question to ask, and this is precisely the question we address.

This paper assumes a discrete-time market, where the investor is only allowed to trade at discrete trading times $\mathbb T:=\{t_i\}_{i=0}^n \subset [0, 1]$. Under certain assumptions on the financial market, such markets will have at least a \textit{risk-neutral measure} or \textit{equivalent martingale measure}: a probability measure under which the discounted risky asset is a martingale. However, except in very restrictive settings such risk-neutral measures are not unique. Feasible prices for a financial derivative are given by expectations of the payoff of the contract under risk-neutral measures. Under no-arbitrage assumptions, bounds on feasible prices are also given by the so-called super-hedging prices (\cite[Theorem 1.4]{acciaio}). It turns out that under certain assumptions on the market, both pricing approaches coincide -- i.e. the supremum of the expected payoff under risk-neutral measures coincides with the infimum super-hedging price (\cite{acciaio, obloj, discrete_modelfree1}). This is called the pricing-hedging duality.

Knowledge of the implied expected signature -- defined as the expected signature that matches the market prices of all observable financial derivatives -- is sufficient to price financial contracts. The expected signature of a stochastic process (\cite{ilya2, expected_sig_bm, ilya, hao}) is an object that plays a similar role to the moments of an $\mathbb R^d$-valued random variable. Under certain assumptions, the expected signature determines the law of the stochastic process (\cite{ilya, ilya2}), just like under some assumptions the moments of the $\mathbb R^d$-valued random variable will determine the law of the random variable. Moreover, knowing the implied expected signature is equivalent to knowing the prices of all signature payoffs. Therefore, the fact that the expected signature determines the law of the market dynamics suggests that the prices of all signature payoffs determine the market dynamics. Expected signatures have also been used in other areas of mathematical finance, such as in optimal execution \cite{optimal1, optimal2}.

The remainder of the paper is structured as follows. In Section \ref{sec:signature} we introduce the notion of the signature of a path, a crucial object in this paper. Section \ref{sec:framework} introduces the financial framework we will consider. In Section \ref{sec:pricing} we introduce signature payoffs and we show how they are used to numerically estimate prices of other exotic derivatives. Finally, in Section \ref{sec:numerical experiment} we show the feasibility of our numerical method with numerical experiments. All proofs have been postponed to Appendix \ref{appendix:proofs}.



\section{Signature of a path}\label{sec:signature}

Consider the $d$-dimensional Euclidean space $\mathbb R^d$. Denote by $(\mathbb R^d)^{\otimes n} := \underbrace{\mathbb R^d \otimes \ldots \otimes \mathbb R^d}_{n}$, where $\otimes$ is the tensor product. We define the \textit{extended tensor algebra} over $\mathbb R^d$ as
$$T((\mathbb R^d)) := \{(a_0, a_1, \ldots, a_n, \ldots) : a_n \in (\mathbb R^d)^{\otimes n}\mbox{ for all }n\geq 0\}.$$ Similarly, we denote \begin{align}\label{eq:ta}
&T(\mathbb R^d) := \{ (a_0, a_1, \ldots, a_n, \ldots) \in T((\mathbb R^d)) : \exists N\in \mathbb N \mbox{ s.t. }a_n=0\mbox{ for all }n\geq N\},\\
&T^N(\mathbb R^d) := \{ (a_0, a_1, \ldots, a_n, \ldots) \in T((\mathbb R^d)) : a_n=0 \mbox{ for all }n\geq N\}
\end{align} which we call the \textit{tensor algebra} and \textit{truncated tensor algebra at $N\in \mathbb N$}, respectively. All three are algebras with the product $\otimes$ and sum $+$.

Denote by $(\mathbb R^d)^\ast$ the dual space of $\mathbb R^d$. Define the projection $\mathbf{1}^\ast\in T((\mathbb R^d)^\ast)$ by $\langle \mathbf 1^\ast, \mathbf a\rangle := a_0$ for all $\mathbf a = (a_0, a_1, \ldots)\in T((\mathbb R^d))$. Let $\{e_1, \ldots, e_d\}$ be a basis for $\mathbb R^d$, and $\{e_1^\ast, \ldots, e_d^\ast\}$ the corresponding dual basis for $(\mathbb R^d)^\ast$. Then, $$\{e_I^\ast := e_{i_1}^\ast \otimes \ldots \otimes e_{i_n}^\ast \;|\; I=(i_1, \ldots, i_n) \in \{1, \ldots, d\}^n, n\in \mathbb N\}\cup \{\mathbf{1}^\ast\}$$ is a basis for $T((\mathbb R^d))^\ast \cong T((\mathbb R^d)^\ast)$. There is a certain product that can be defined on $T((\mathbb R^d)^\ast)$, called the \textit{shuffle product}:

\begin{definition}\label{def:shuffle}
Let $I=(i_1, \ldots, i_n)\in \{1, \ldots, d\}^n$, $J=(j_1, \ldots, j_m)\in \{1, \ldots, d\}^m$. The shuffle product of $e_I^\ast$ and $e_J^\ast$, denoted by $\phantom{ }\shuffle$, is defined inductively by:
$$e_I^\ast \shuffle\; e_J^\ast:= (e_{(i_1, \ldots, i_{n-1})}^\ast \shuffle\; e_J^\ast)\otimes e_{i_n}^\ast + (e_I^\ast \shuffle\; e_{(j_1, \ldots, j_{m-1})}^\ast) \otimes e_{j_m}^\ast,$$ and $e_I^\ast \shuffle\; \mathbf{1}^\ast = \mathbf{1}^\ast \shuffle\; e_I^\ast = e_I^\ast$. The operation is extended by bilinearity to $T((\mathbb R^d)^\ast)$.
\end{definition}

The signature of a path with bounded variation is a $T((\mathbb R^d))$-valued object, defined as follows:

\begin{definition}[Signature of a path]

Given a path $Z\in C([0, T]; \mathbb R^d)$ with bounded variation, we denote its signature on $[s, t] \subset [0, T]$ by

$$\Sig(Z)_{s,t} := (1, Z^1_{s,t}, Z^2_{s,t}, \ldots, Z^n_{s,t}, \ldots) \in T((\mathbb R^d))$$ where
\begin{equation}\label{eq:iterated integral}
Z^n_{s,t} := \int_{s<u_1<\ldots<u_n<t}dZ_{u_1}\otimes \ldots \otimes dZ_{u_n} \in (\mathbb R^d)^{\otimes n}.
\end{equation}
Similarly, the truncated signature of order $N\in \mathbb N$ is denoted by

$$\Sig^N(Z)_{s,t} := (1, Z^1_{s,t}, \ldots, Z^N_{s,t})\in T^N(\mathbb R^d).$$ When we refer to the \textit{signature of $Z$} without specifying the interval $[s, t]$, we will implicitly refer to the signature on $[0, T]$, i.e. $\Sig(Z)_{0,T}$.
\end{definition}

A thorough study of the signature of a path is beyond the scope of this paper -- we refer the reader to \cite{lyonsbook} for a more detailed view of signatures and rough paths. We state the following key result; it guarantees that linear functions on the signature form an algebra.

\begin{proposition}{[\cite[Theorem 2.15]{lyonsbook}]}\label{prop:shuffle}
Let $Z:[0, T] \to \mathbb R^d$ be a continuous path with bounded variation. Let $\ell_1, \ell_2\in T((\mathbb R^d)^\ast)$ be two linear functionals. Then,

$$\langle \ell_1, \Sig(Z)_{s,t}\rangle \langle \ell_2, \Sig(Z)_{s,t}\rangle = \langle \ell_1 \shuffle\; \ell_2, \Sig(Z)_{s,t}\rangle \quad \forall\, 0\leq s < t \leq T.$$
\end{proposition}

\section{Framework}\label{sec:framework}

\subsection{The market}

Let $T>0$ and let $\mathbb T=\{t_i\}_{i=0}^n \subset [0, 1]$ with $0=t_0<\ldots<t_n=1$ be the trading times. We assume that the market consists of a single risky asset, although all results generalise to the multi-asset case. The space of (discounted) price paths is defined as $\Omega:=\{X:\mathbb T \to \mathbb R_+ : X_0 = 1\}$. Therefore, we assume that the initial price is normalised to 1.

Traders make decisions based on past information and the outcome of these decisions is known as new information is revealed. Thus, we transform price paths to incorporate past and future information. The transformation, known as the \textit{lead-lag transformation} of a path, is studied in \cite{guy} and is defined below.

\begin{definition}[Lead-lag transformation]\label{def:leadlag}
Given a price path $X\in \Omega$, define its lead-lag transformation $\widehat X:[0,1]\to \mathbb R^2 \oplus \mathbb R$ by the continuous path given by
\begin{align*}
&\widehat{X}_{2k/2n} := \big ( (t_k, X_{t_k}), X_{t_k}\big )\in \mathbb R^2\oplus \mathbb R,\\
&\widehat{X}_{(2k+1)/2n} := \big ( (t_k, X_{t_k}), X_{t_{k+1}}\big )\in \mathbb R^2\oplus \mathbb R,
\end{align*} and linear interpolation in between. For each $t\in [0, 1]$, we write $\widehat X_t = (X_t^b, X_t^f)$ where $X_t^b\in \mathbb R^2$ denotes the \textit{lag} (backward) component and $X_t^f\in \mathbb R$ denotes the \textit{lead} (forward) component. We define the space of lead-lag price paths $\widehat \Omega := \{ \widehat X:X\in \Omega\}$. Paths in $\widehat \Omega$ are continuous and have bounded variation, and hence $\widehat \Omega$ is a subspace of $BV([0, 1]; \mathbb R^2\oplus \mathbb R)$, so that $\widehat \Omega$ is equipped with the norm $\lVert \cdot \rVert_{BV}$. Similarly, for $t\in [0, 1]$ we define $\widehat \Omega_t := \{ \widehat X|_{[0, t]} : \widehat X \in \widehat \Omega\}$. We denote by $\pi^b$ and $\pi^f$ the projections $\pi^b:\mathbb R^2\oplus \mathbb R\to \mathbb R^2$ and $\pi^f:\mathbb R^2\oplus \mathbb R\to \mathbb R$ onto the lag and lead components, respectively.
\end{definition}

\begin{figure}
\includegraphics[width=\linewidth]{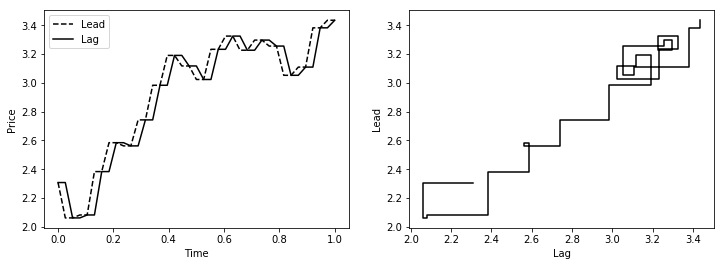}
\caption{Lead-lag transformation of a price path. The figure on the left shows the lead and lag components of the path, and the figure on the right shows the lag component plotted against the lead component.}
\label{fig:leadlag}
\end{figure}

Figure \ref{fig:leadlag} shows the lead-lag transformation of a certain price path. As the name suggests, the lead component (the future) is \textit{leading} the lag component (the past).

Notice that because all paths in $\widehat \Omega$ are piecewise linear, $(\widehat \Omega, \lVert \cdot \rVert_{BV})$ is separable. We consider the measure space $(\widehat \Omega, \mathcal B(\widehat \Omega))$.

It turns out that the signature of lead-lag transformed paths uniquely determines the path:

\begin{lemma}[Uniqueness of signature, \cite{uniqueness}]\label{lemma:uniqueness}
The signature map $\Sig:\widehat \Omega \to T((\mathbb R^2\oplus \mathbb R))$ is injective.
\end{lemma}

This result is not true if we consider the signature map $\Sig:\Omega \to T((\mathbb R))$. In other words, given a path $X\in \Omega$, its signature $\Sig(X)$ does not uniquely determine $X$ itself, but if its lead-lag transformation $\widehat X\in \widehat \Omega$ is considered, then $\Sig(\widehat X)$ \textit{does} determine uniquely $\widehat X$ (and hence $X$). This is the main reason why lead-lag transformed paths are considered in this paper.

\subsection{Payoffs and no-arbitrage assumptions}

Here, we consider exotic, path-dependent payoffs. More specifically, we define payoffs as follows:

\begin{definition}[Payoff] A payoff is a Borel-measurable function $G:\widehat \Omega \to \mathbb R$.
\end{definition}

The following are the assumptions we impose on the market.

\begin{assumption}\label{assumption}
Let $\mathcal M$ denote the set of all risk-neutral measures. We assume that the market satisfies the following assumptions:
\begin{enumerate}
\item There is a finite family of payoffs $\mathfrak F =\{F_i\}_{i}$ that can be bought at $t=0$, where $F_i:\widehat \Omega \to \mathbb R$.
\item There are no market frictions such as transaction costs,  slippage and bid-ask spread.
\item $\mathcal M \neq \varnothing$.
\item $\mathcal M$ is tight.
\item $\mathbb E^{\mathbb Q} [\Sig^{N}(\widehat X)_{0,1}]$ is well-defined for all $\mathbb Q\in \mathcal M$ and $N\in \mathbb N$.
\end{enumerate}
\end{assumption}

The assumption that $\mathbb E^{\mathbb Q} [\Sig^{N}(\widehat X)_{0,1}]$ is well-defined for all $\mathbb Q\in \mathcal M$ and $N\in \mathbb N$ is a mild, nonrestrictive assumption that is considered for technical reasons. It is a path-space analogous of the ``moments of all order exist'' assumption for finite dimensional random variables.

\begin{definition}
We denote by $\mathcal P$ the pricing map, given by $\mathcal{P}(G) := \sup_{\mathbb  Q \in \mathcal M} \mathbb E^{\mathbb Q} [G(\widehat X)]$ for each payoff $G$.
\end{definition}

Let $F\in \mathfrak F$. Notice that the mapping $\mathcal M\to \mathbb R$ given by $\mathbb Q\mapsto \mathbb E^{\mathbb Q}[F(\widehat X)]$ is constant, so that $\mathcal P (F)=\mathbb E^{\mathbb Q}[ F(\widehat X)]$ for any fixed $\mathbb Q\in \mathcal M$.

\subsection{Trading strategies}

Trading strategies are investment policies that determine how much of the asset should be held at each time. Moreover, the decision has to be made based on the information about the past only. In our framework, this means that a trading strategy is a real-valued function acting on the lag component of the lead-lag transformation. We make this more precise by introducing the space $\Lambda$ -- see below. A similar space is discussed in \cite{rama2, rama1, bookstopped, rama3} and in \cite{dupire, promel, riga} in the context of finance.

\begin{definition}[Trading strategies]
We denote $\Lambda := \bigcup_{t\in [0, 1]} \pi^b(\widehat\Omega_t)$, where $\pi^b$ is the projection onto the lag component defined in Definition \ref{def:leadlag}. This is a metric space with a certain metric (\cite[Section 2.2]{rama2}). The space of trading strategies is defined as $\mathcal T := C(\Lambda; \mathbb R)$.
\end{definition}

Intuitively, the space $\mathcal T$ consists of all non-anticipative processes; trading should only be done with past information because no information about the future can be used.

\begin{remark}
The space $\Lambda$ is a Polish space (see \cite{bookstopped}).
\end{remark}

The outcome of a trading strategy $\theta\in \mathcal T$ depends on the trajectory of the price path $\widehat X\in \widehat \Omega$. The gains of the trading strategy $\theta$ are defined below.

\begin{definition}
For $\theta \in \mathcal T$ and $X\in \widehat \Omega$, the gains of the trading strategy $\theta$ are given by

$$(\theta \bullet  \widehat X)_{\mathbb T} := \sum_k \theta(X^b|_{[0, t_k]})(X_{t_{k+1}} - X_{t_k}).$$
\end{definition}

In this paper we allow semi-static and dynamic hedging. Recall that $\mathfrak{F}$ is the family of payoffs that can be bought at $t=0$ (see Assumption \ref{assumption}). The performance of the hedging strategy will be determined by three factors: the price that is paid for each of the payoffs in $\mathfrak F$, the gains or losses that stem from holding the payoffs in $\mathfrak F$, and the gains that come from dynamically trading the underlying asset.

\begin{definition}\label{def:hedging strategy}
A hedging strategy is a pair $h=((a_F)_{F\in \mathfrak F}, \theta)$, where $a_F\in \mathbb R$ for all $F\in \mathfrak F$ and $\theta \in \mathcal T$. The Profit and Loss (P\& L) of $h$ is defined by
\begin{align*}
V_h: \widehat \Omega &\to \mathbb R\\
\widehat X &\mapsto \sum_{F\in \mathfrak F} a_F (F(\widehat X) - \mathcal P(F)) + (\theta \bullet X)_{\mathbb T}.
\end{align*}
We denote by $\mathcal H$ the space of hedging strategies. A payoff $G:\widehat \Omega\to \mathbb R$ is said to be attainable if there exists a hedging strategy $h\in \mathcal H$ such that $G\equiv V_h$ on $\widehat \Omega$. We denote the space of all attainable payoffs by $\mathcal A$. Similarly, we define, for $\varepsilon> 0 $, $\mathcal A_\varepsilon := \{G:\widehat \Omega \to \mathbb R \mbox{ payoff } |\; \exists H\in \mathcal A\mbox{ such that }\lVert G - H \rVert_{L^\infty(\widehat \Omega)}<\varepsilon\}$ the space of almost attainable payoffs.
\end{definition}

\section{Signature payoffs}\label{sec:pricing}

We will now define a special class of continuous payoffs -- the family of signature payoffs. These were first introduced in \cite{sina, imanol}, and as Arrow-Debreu securities \cite{arrow, debreu}, they are \textit{primitive securities} in the sense that path-dependent exotic derivatives can be approximated by these signature payoffs.

\begin{definition}[Signature payoff]
Let $\ell \in T((\mathbb R^2\oplus \mathbb R)^\ast)$. The signature payoff $\mathcal S^\ell$ is defined by
\begin{align*}
\mathcal{S}^\ell : \widehat {\Omega} & \to \mathbb R\\
\widehat X &\mapsto \langle \ell, \Sig(\widehat X)_{0, 1}\rangle.
\end{align*}
\end{definition}

In other words, a signature payoff is a financial derivative whose payoff is given as a linear combination of certain iterated integrals. Because we are using iterated integrals against the price path, the family of all signature payoffs effectively contains all possible dynamic hedging strategies.

The following Proposition, see \cite{sina}, shows that signature payoffs are primitive securities. Indeed, the proposition establishes that continuous payoffs can be locally approximated by signature payoffs:

\begin{proposition}\label{prop:payoffs local}
Let $G:\widehat \Omega\to \mathbb R$ be a continuous payoff, and let $\mathcal K\subset \widehat \Omega$ be compact. Given any $\varepsilon > 0$, there exists a signature payoff $\mathcal S^\ell$ with $\ell\in T((\mathbb R^2\oplus \mathbb R)^\ast)$ such that
$$|G(\widehat X)-\mathcal S^\ell(\widehat X)|<\varepsilon\quad \forall \,\widehat X \in \mathcal K.$$
\end{proposition}

We have the following immediate corollary.

\begin{corollary}\label{cor:payoffs local}
Let $G:\widehat \Omega \to \mathbb R$ be a continuous payoff. Let $\varepsilon > 0$. Then, there exists a compact $\mathcal K_\varepsilon\subset \widehat \Omega$ and a signature payoff $\mathcal S^\ell$ with $\ell\in T((\mathbb R^2\oplus \mathbb R)^\ast)$ such that
\begin{enumerate}
\item $\mathbb Q(\widehat X\in \mathcal K_\varepsilon) > 1-\varepsilon$ for all risk-neutral measures $\mathbb Q\in \mathcal M$,
\item $|G(\widehat X)-\mathcal S^\ell(\widehat X)|<\varepsilon$ for all $\widehat X \in \mathcal K_\varepsilon$.
\end{enumerate}
\end{corollary}

The corollary states that there exists a \textit{large} compact set -- large in the sense that with very high probability, the price path will belong to the compact set -- such that on the compact set the exotic derivative is very close to a signature payoff.

Moreover, prices of payoffs can be well-approximated using signature payoffs:

\begin{proposition}\label{prop:approximate price}
Set $\varepsilon > 0$. Let $G:\widehat \Omega\to \mathbb R$ be a $\mathbb Q$-integrable payoff for all $\mathbb Q\in \mathcal M$. Assume $G$ is either in $\mathcal A_{\varepsilon/4}$ (i.e. almost attainable) or bounded. Then, there exists a compact $\mathcal K_\varepsilon \subset \widehat \Omega$ and a signature payoff $\mathcal S^\ell$ such that, for $\mathcal L := \mathds 1_{\mathcal{K}_\varepsilon} \mathcal S^{\ell}:\widehat \Omega\to\mathbb R$,

\begin{enumerate}
\item $\mathbb Q(\widehat X\in \mathcal K_\varepsilon) > 1-\varepsilon$ for all $\mathbb Q\in \mathcal M$.
\item $\mathcal P(|G-\mathcal L|) = \sup_{\mathbb Q\in \mathcal M} \mathbb E^{\mathbb Q} [|G(\widehat X)-\mathcal{L}(\widehat X)|] < \varepsilon$.
\end{enumerate}

\end{proposition}

\subsection{The implied expected signature}\label{subsec:implied}

Let $\varepsilon > 0$. By Proposition \ref{prop:approximate price}, we may pick a compact set $\mathcal K_\varepsilon\subset \widehat \Omega$ and a family of signature payoffs $\mathscr L = \{ \mathcal S^{\ell_F}\}_{F\in \mathfrak F}$ with $\ell_F \in T((\mathbb R^2\oplus \mathbb R)^\ast)$, such that:

\begin{enumerate}
\item $\mathbb Q(\widehat X\in \mathcal{K}_\varepsilon) > 1 - \varepsilon$ for all $\mathbb Q\in \mathcal M$,
\item $\mathcal P(|F-\mathds{1}_{\mathcal K_\varepsilon}\mathcal S^{\ell_F}|)<\varepsilon$ for all $F\in \mathfrak F$.
\end{enumerate}

Fix a risk-neutral measure $\mathbb Q\in \mathcal M$ now. We then have:

\begin{equation}\label{eq:price approximation}
\langle \ell_F, \mathbb E^{\mathbb Q}[\Sig(\widehat X)_{0,1} \;;\; \mathcal K_\varepsilon] \rangle \approx \mathcal P(F)\quad \forall F\in \mathfrak F,
\end{equation}
where $\mathbb E^{\mathbb Q}[\Sig(\widehat X)_{0,1} \;;\; \mathcal K_\varepsilon]$ denotes $\mathbb E^{\mathbb Q}[\Sig(\widehat X)_{0,1} \mathds 1_{\mathcal K_\varepsilon}]$. Choose $N\in \mathbb N$ sufficiently large so that $\ell_F\in T^N(\mathbb R^2\oplus \mathbb R)$ for all $F\in \mathfrak F$, which is possible because $\mathfrak F$ is finite and by definition of tensor algebra \eqref{eq:ta}. Then, \eqref{eq:price approximation} becomes:

\begin{equation}\label{eq:finite price approximation}
\langle \ell_F, \mathbb E^{\mathbb Q}[\Sig^N(\widehat X)_{0,1} \;;\; \mathcal K_\varepsilon] \rangle \approx \mathcal P(F)\quad \forall F\in \mathfrak F.
\end{equation} If the risk-neutral measures are unknown and we only observe market prices for the exotic payoffs $\mathfrak F$, we may look for the expected signature that best matches the observed prices -- i.e. the best expected signature fit for the approximation \eqref{eq:finite price approximation}. This can be achieved using linear regression. The complexity of the estimation will, in general, depend on the dimension of $T^N(\mathbb R^2 \oplus \mathbb R)$ -- which is $(3^{N+1} - 1)/2$ -- and the number of exotics options $|\mathfrak{F}|$. Notice, in particular, that the complexity grows exponentially with $N$. In the experiments in Section \ref{sec:numerical experiment}, however, it sufficed to set $N=5$.

The extent to which we can accurately estimate the expected signature depends on the invertibility  of the linear map \begin{align}\label{eq:linear map invert}
T^N(\mathbb R^2\oplus \mathbb R) &\to \mathbb R^{|\mathfrak F|}\notag \\
\mathbf a &\mapsto (\langle \ell_F, \mathbf a\rangle)_{F\in \mathfrak F}.
\end{align}
If the family of exotic payoffs $\mathfrak F$ is rich enough (i.e. the linear map \eqref{eq:linear map invert} has high rank) we may be able to estimate the expected signature that best matches the observed prices for such exotic payoffs. We call this expected signature the \textit{implied expected signature}, in analogy with the concept of \textit{implied volatility}.

If we are given a (potentially illiquid) payoff $G:\widehat \Omega \to \mathbb R$ that is close to the space of attainable payoffs $\mathcal A$ (Definition \ref{def:hedging strategy}) and that satisfies the assumptions of Proposition \ref{prop:approximate price}, we may estimate its price by $\mathcal P (G) \approx \langle \ell_G, \mathbf E\rangle$ where $\ell_G$ is given by Proposition \ref{prop:approximate price} and $\mathbf E$ is the implied expected signature that was estimated from market data.

This approach provides a numerical method to estimate prices of certain exotic payoffs from observed prices of other exotics in a model-free manner. The effectiveness and accuracy of the method depends on the richness of the payoffs $\mathfrak F$ with observable prices.

\begin{example}
In what follows, we discuss a couple of examples to provide intuition of the meaning of some of the terms of the implied expected signature.

Let $\widehat X=(X^b, X^f)\in \widehat \Omega$ and write $\widehat X_t = ((\widehat X_t^1, \widehat X_t^2), \widehat X_t^3)\in \mathbb R^2\oplus \mathbb R$. By the definition of the iterated integrals \eqref{eq:iterated integral} that form the signature, the first term of the signature is given by:
$$\langle e_2^\ast,\, \Sig(\widehat X)_{0,1}\rangle \stackrel{\mathclap{\normalfont\small{\text{def}}}}{=} \int_0^1 d\widehat X_t^{1}=X_1-X_0.$$
In other words, the first term of the signature is given by the increment of the price path. Hence, the first term of the implied expected signature $\langle e_2^\ast, \mathbf E\rangle$ is an estimate of the forward price of the underlying asset.

For the second term of the signature, the lead-lag transformation $\widehat X$ has the property that
$$\left \langle e_3^\ast \otimes e_2^\ast - e_2^\ast \otimes e_3^\ast, \Sig(\widehat X)_{0,1}\right \rangle \stackrel{\mathclap{\normalfont\small{\text{def}}}}{=}   \int_0^1\int_0^t d\widehat X^3_s d\widehat X^2_t - \int_0^1\int_0^t d\widehat X^2_s d\widehat X^3_t  = \langle X \rangle,$$
where $\langle X \rangle=\sum_{i=1}^n (X_{t_{i-1}} - X_{t_i})^2$ is the quadratic variation -- or volatility -- of the price path. Therefore, the second term of the implied expected signature $\langle e_3^\ast \otimes e_2^\ast - e_2^\ast \otimes e_3^\ast, \mathbf E\rangle$ is an estimate of the expected volatility of the market.

In a nutshell, the first two order terms of the implied expected signature contain information about the forward price and expected volatility of the market. Higher order terms capture other aspects of the market dynamics. In fact, the full implied expected signature (when $N\to \infty$) characterises the market dynamics, see \cite{ilya}.
\end{example}

\section{Numerical experiment}\label{sec:numerical experiment}

\subsection{Approximating payoffs by signature payoffs}\label{subsec:approx}

In this section we discuss how to approximate a payoff $G:\widehat \Omega \to \mathbb R$ by a signature payoff $\mathcal S^{\ell}:\widehat \Omega\to \mathbb R$.

We begin by choosing a dataset of paths $\{\widehat X^{(i)}\}_{i=1}^n\subset \widehat \Omega$. In the experiments in this section, we generated 10,000 paths from a Black--Scholes model with a constant volatility ranging from 5\% to 40\% -- as we will see, the approximating signature payoff generalised well even to non-Black--Scholes paths.

Once the set of paths is chosen, we compute the signatures $\{\Sig^N(\widehat X^{(i)})_{0, 1}\}_{i=1}^n$ of order $N\in \mathbb N$. In the following experiments we fixed $N=5$. There are multiple publicly available, open source libraries to compute signatures, such as \texttt{esig}\footnote{\url{https://pypi.org/project/esig/}} or \texttt{iisignature}\footnote{\url{https://github.com/bottler/iisignature}}.

Finally, one can apply a linear regression of the signatures $\{\Sig^N(\widehat X^{(i)})_{0, 1}\}_{i=1}^n$ against the payoff cash flows $\{G(\widehat X^{(i)})\}_{i=1}^n$ to find a linear functional $\ell \in T^N(\mathbb R^2\oplus \mathbb R)$ such that
$$\langle \ell, \Sig^N(\widehat X^{(i)})_{0,1}\rangle \approx G(\widehat X^{(i)}) \quad \forall \, i=1, \ldots, n.$$

\subsection{Pricing from market data}\label{subsec:pricing experiment}

\begin{figure}
\centering
\begin{subfigure}[b]{0.9\textwidth}
   \includegraphics[width=1\linewidth]{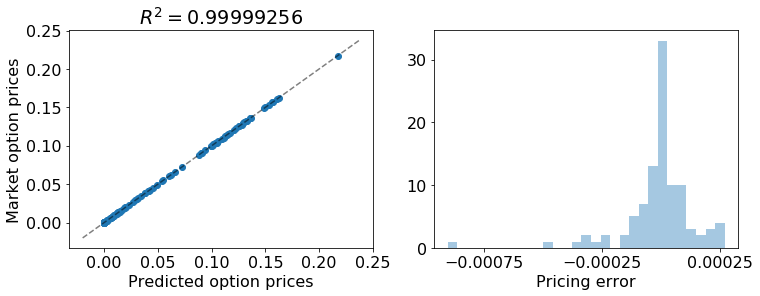}
   \caption{Hull--White model.}
   \vspace{0.5cm}
\end{subfigure}

\begin{subfigure}[b]{0.9\textwidth}
   \includegraphics[width=1\linewidth]{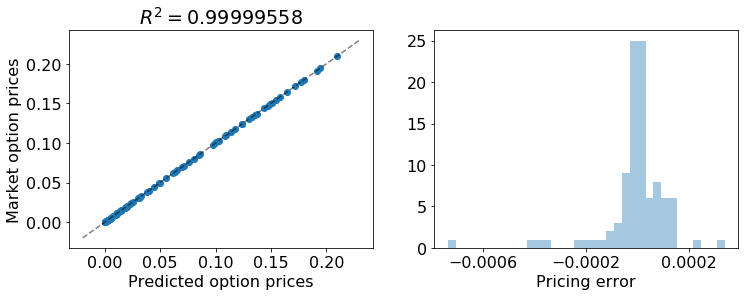}
   \caption{GARCH model.}
   \vspace{0.5cm}
\end{subfigure}

\begin{subfigure}[b]{0.9\textwidth}
   \includegraphics[width=1\linewidth]{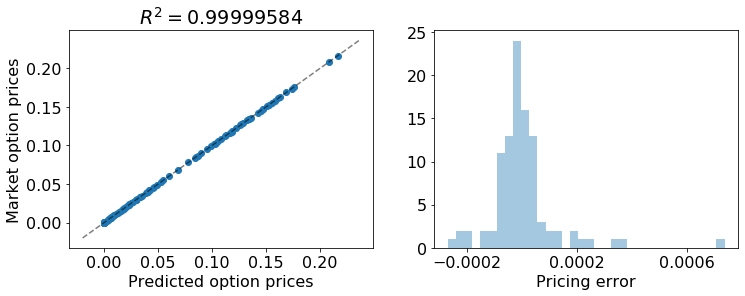}
   \caption{Rough volatility model.}
   \vspace{0.5cm}
\end{subfigure}

\caption{Accuracy of the estimated market prices for different market dynamics. The scatter plot on the left shows the predicted option price (horizontal axis) against the actual market price (vertical axis). The histogram on the right shows the pricing error.}
\label{fig:pricing}
\end{figure}


In this section we implement the proposed approach to estimate the implied expected signature from a family of exotic options $\mathfrak F$, which we employ to price other exotic options. The family $\mathfrak F$ of exotic options available to buy at $t=0$ (see Assumption \ref{assumption}) that was considered consists of 100 exotic and vanilla payoffs with maturity 1 year. More specifically, we consider European options, barrier up-and-out options, barrier up-and-in options and variance options (25 payoffs of each type). The size of the dataset, as well as the payoff types, is similar to that offered by consensus market price providers. The discrete timeline $\mathbb T$ we consider is a timeline consisting of all trading days in that year.

To generate prices of $\mathfrak F$, we considered prices generated by three stochastic volatility models: Hull--White model \cite{hullwhite}, GARCH model \cite{garch1, garch2} and the rough volatility model \cite{roughvol}. In all cases the market dynamics are unknown to the trader because the trader has access to the prices of exotics in $\mathfrak F$ but has no knowledge about the market dynamics that generated those prices. Therefore, the only information the trader may use is the market prices of $\mathfrak F$.

We used the procedure described in Section \ref{subsec:approx} to approximate each payoff in $\mathfrak F$ by a signature payoff, and then we followed the numerical method proposed in Section \ref{subsec:implied} to estimate the implied expected signature. Recall that signatures of order $N=5$ are considered.

We generated prices of 100 new exotic derivatives (different from those in $\mathfrak F$) using each of the three models (i.e. Hull--White model, GARCH and rough volatility). We used the implied expected signature to predict the prices of these 100 new exotic payoffs, for which we obtain accurate results for all models -- see Figure \ref{fig:pricing}. The $R^2$ of the predictions is very high for the three models.

\begin{figure}
    \centering
    \includegraphics[width=0.45\linewidth]{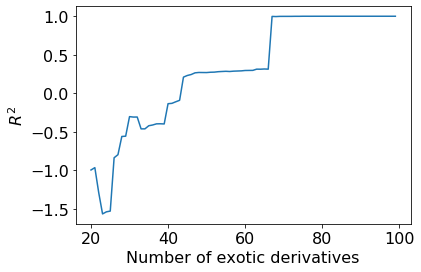}
    \includegraphics[width=0.45\linewidth]{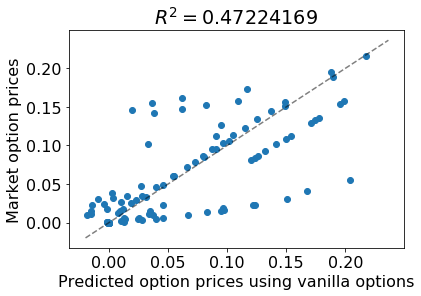}
    \caption{On the left, $R^2$ of the predicted market prices of the exotic functions, as a function of the size of available market prices $|\mathfrak F|$. On the right, the predicted market prices of exotic options (horizontal axis) against the actual market price (vertical axis) when $\mathfrak F$ is only formed by vanilla options. The $R^2$ is very poor when only vanilla options are used. Both plots were obtained for prices computed under the Hull--White model.}
    \label{fig:richness of F}
\end{figure}

The effectiveness of the methodology depends on how rich the class of exotic payoffs $\mathfrak F$ is, as discussed in Section \ref{subsec:implied}. For instance, the larger the family of exotics $\mathfrak F$ is, the more accurate the estimation of the implied expected signature is, which leads to more accurate market price predictions. This is shown on the first figure of Figure \ref{fig:richness of F}, where the $R^2$ of the predicted market prices of exotic payoffs is shown as a function of the size of the available exotic payoffs $\mathfrak F$. As we see, the performance is not very good when the size of $\mathfrak F$ is small, but it rapidly improves as the size of $\mathfrak F$ (i.e. the set of traded derivatives) increases. On the other hand, the types of exotic payoffs that are considered also matters. For instance, if $\mathfrak{F}$ only includes vanilla options we have information about the marginal distribution of the price path, but it will not, in general, have information about the transition probabilities of prices. Therefore, vanilla options are not sufficient to price other exotic options. This is illustrated on the second figure of Figure \ref{fig:richness of F}, where the accuracy of predicted market prices of exotic payoffs is shown in the case where $\mathfrak F$ only contains vanilla options.

Table \ref{table:table1} includes the accuracy of the predicted prices, measured with different metrics, for different choices of $\mathfrak{F}$. As the table indicates, a large family of payoffs $\mathfrak{F}$ that includes vanilla and exotic options leads to the best performance.

\begin{table}[]
\centering
\scriptsize
\begin{tabular}{l|c|ccc}
\hline
\hline
  \noalign{\vskip 2mm}
           & \textbf{Vanillas only in $\mathfrak{F}$} & \multicolumn{3}{c}{\textbf{Vanilla and exotic options in $\mathfrak F$}} \\  
 & $|\mathfrak F|=100$ &  $|\mathfrak F|= 33$   &  $|\mathfrak F|=66$   &   $|\mathfrak F|=100$   \\
           
\noalign{\vskip 2mm}
\hline
  \noalign{\vskip 2mm}
\textbf{$R^2$ of predictions} & 0.472242 &   -0.461247  & 0.313326  & \textbf{0.999993}     \\
\textbf{Mean-squared error} & $1.88\times 10^{-3}$ & $1.81 \times 10^{-2}$     & $1.54\times 10^{-6}$     & $\mathbf{1.48\times 10^{-8}}$     \\
\textbf{Mean absolute error} & $2.94\times 10^{-2}$ & $7.89 \times 10^{-2}$     & $6.04 \times 10^{-4}$     & $\mathbf{3.39 \times 10^{-5}}$     \\   \noalign{\vskip 2mm}
\hline\hline
\end{tabular}
\caption{Comparison between different choices of available payoffs $\mathfrak{F}$. In the first column only 100 vanilla options are in $\mathfrak{F}$ and in the second column, $\mathfrak{F}$ consists of both vanilla and exotic options (barrier up-and-out options, barrier up-and-in options and variance option). The table suggests that a large family of payoffs $\mathfrak{F}$ that includes exotic options leads to the best performance.} \label{table:table1}
\end{table}

\section{Conclusion}

In this paper we proposed a numerical method for estimating prices of exotic derivatives from available market prices of other exotic derivatives. We do so in a model-free setting, where only minimal assumptions about the market are made.

Our methodology extends the ideas from the contemporaneous work in \cite{sina}. We show that from market prices of exotic derivatives we can estimate a certain non-observable quantity called the \textit{implied expected signature}. This object consists of a collection of statistics that characterise the market dynamics, and it can be used to numerically estimate market prices of other exotic derivatives.

We believe there are several directions that this paper can be extended in. For instance, core ideas from the paper could be applied to risk management of large portfolios of exotic derivatives: portfolio optimisation, value at risk (VaR) calculation, etc. Further research in this direction may lead to new tools in these fields.

\appendix
\section{Proofs}\label{appendix:proofs}

\begin{proof}[Proof of Proposition \ref{prop:payoffs local}]
First, we show that the class of all signature payoffs forms an algebra. Let $\ell_1,\ell_2\in T((\mathbb R^2\oplus \mathbb R)^\ast)$, and consider the corresponding signature payoffs $\mathcal S^{\ell_1}, \mathcal S^{\ell_2}:\widehat \Omega\to \mathbb R$. It is clear that $\mathcal S^{\ell_1} + \mathcal S^{\ell_2}=\mathcal S^{\ell_1+\ell_2}$ -- in other words, the sum of two signature payoffs is a signature payoff. Moreover, by the shuffle product property (see Proposition \ref{prop:shuffle}), we have:
$$\mathcal S^{\ell_1}\cdot \mathcal S^{\ell_2} = \mathcal S^{\ell_1\shuffle\ell_2}.$$
That is, the product of two signature payoffs is another signature payoff, and the linear functional is given by the shuffle product of the linear functionals of the original signature payoffs. Therefore, the space of signature payoffs forms an algebra.

By Lemma \ref{lemma:uniqueness}, the signature map $\widehat \Omega \to T((\mathbb R^2\oplus \mathbb R))$ is injective. Therefore, given $\widehat X_1,\widehat X_2\in \widehat \Omega$ distinct, we have $\Sig(\widehat X_1) \neq \Sig(\widehat X_2)$. It follows immediately that there exists a signature payoff $\mathcal S^\ell$, with $\ell \in T((\mathbb R^2\oplus \mathbb R)^\ast)$, such that $\mathcal S^\ell (\widehat X_1) \neq \mathcal S^\ell (\widehat X_2)$. Therefore, signature payoffs separate points.

Given that the space of signature payoffs trivially contain constants, we conclude by the Stone--Weierstrass theorem that given any continuous payoff $G:\widehat \Omega\to \mathbb R$ and a compact set $\mathcal K\subset \widehat \Omega$, there exists a signature payoff $\mathcal S^\ell$ with $\ell\in T((\mathbb R^2\oplus \mathbb R)^\ast)$ such that $|G(\widehat X) - \mathcal S^\ell (\widehat X)|<\varepsilon$ for all $\widehat X\in \mathcal K$.
\end{proof}

\begin{proof}[Proof of Corollary \ref{cor:payoffs local}]
By tightness of $\mathcal M$, we pick a compact set $\mathcal K_\varepsilon \subset \widehat \Omega$ such that $\mathbb Q(\widehat X\in \mathcal K_\varepsilon) > 1-\varepsilon$ for all risk-neutral measures $\mathbb Q\in \mathcal M$. Apply Proposition \ref{prop:payoffs local} with the payoff $G$ and compact $\mathcal K_\varepsilon$, and conclude that $\lVert G - \mathcal S^\ell \rVert_{L^\infty(\mathcal K_\varepsilon)}<\varepsilon$.
\end{proof}

\vspace{0.2cm}

\begin{proof}[Proof of Proposition \ref{prop:approximate price}]
Substitute $G$ with a continuous payoff that is close to $G$ in $L^1$ if necessary, and assume, without loss of generality, that $G$ is continuous.

By Corollary \ref{cor:payoffs local}, choose a compact set $\mathcal{K}_1\subset \widehat \Omega$ such that $\mathbb Q(\widehat X\in \mathcal K_1) > 1-\varepsilon$ for all $\mathbb Q\in \mathcal M$ and  $\lVert G - \mathcal S^\ell \rVert_{L^\infty(\mathcal K_1)}<\varepsilon/2$.

Moreover, pick $\mathcal K_2\subset \widehat \Omega$ compact such that $\mathbb E^{\mathbb Q}[|G(\widehat X)| \; ; \; \mathcal K_2^c|] < \varepsilon/2$. Indeed:

\begin{enumerate}
\item If $G$ is bounded, then by tightness of $\mathcal M$ we may pick a compact set $\mathcal K_2\subset \widehat \Omega$ such that $\mathbb E^{\mathbb Q}[|G(\widehat X)| \; ; \; \mathcal K_2^c|] < \varepsilon/2$ for all $\mathbb Q\in \mathcal M$. 
\item Otherwise, if $G\in \mathcal A_{\varepsilon/4}$, there exists by definition an attainable payoff \linebreak$H\in \mathcal A$ such that $\lVert H-G\rVert_{L^\infty(\widehat \Omega)}<\varepsilon/4$. Let $\mathcal K_2 \subset \widehat \Omega$ be a compact set such that $\mathbb E^{\mathbb Q}[|H(\widehat X)| \; ; \; \mathcal K_2^c|] < \varepsilon/4$ for all $\mathbb Q\in \mathcal M$. Then, $\mathbb E^{\mathbb Q}[|G(\widehat X)| \; ; \; \mathcal K_2^c|] < \varepsilon/2$ for all $\mathbb Q\in \mathcal M$.
\end{enumerate}

Set $\mathcal K_\varepsilon := \mathcal K_\varepsilon$. Then, $\mathbb Q(\widehat X\in \mathcal K_\varepsilon)>1-\varepsilon$ and

\begin{align*}
\mathcal P(|G-\mathcal L|) &= \sup_{\mathbb Q\in \mathcal M} \mathbb E^{\mathbb Q}[|G(\widehat X) - \mathcal L(\widehat X)] = \sup_{\mathbb Q\in \mathcal M} (\mathbb E^{\mathbb Q}[|G(\widehat X)| \;;\; \mathcal K_\varepsilon^c] + \mathbb E^{\mathbb Q}[|G(\widehat X) - \mathcal S^{\ell}(\widehat X)| \;;\; \mathcal K_\varepsilon])\\
&<\varepsilon.
\end{align*}

\end{proof}

\section*{Disclosure statement}

Opinions and estimates constitute our judgement as of the date of this Material, are for informational purposes only and are subject to change without notice. This Material is not the product of J.P. Morgans Research Department and therefore, has not been prepared in accordance with legal requirements to promote the independence of research, including but not limited to, the prohibition on the dealing ahead of the dissemination of
investment research. This Material is not intended as research, a recommendation, advice, offer or solicitation for the purchase or sale of any financial product or service, or to be used in any way for evaluating the merits of participating in any transaction. It is not a research report and is not intended as such. Past performance is not indicative of future results. Please consult your own advisors regarding legal, tax, accounting or any other aspects including suitability implications for your particular circumstances. J.P. Morgan disclaims any responsibility or liability whatsoever for the quality, accuracy or completeness of the information herein, and for any reliance on, or use of this material in any way.

Important disclosures at: \url{www.jpmorgan.com/disclosures}.

\end{document}